\documentclass[journal]{IEEEtran}

\ifCLASSINFOpdf
\else
   \usepackage[dvips]{graphicx}
\fi
\usepackage{url}
\usepackage{amsmath}
\usepackage{graphicx}
\usepackage{placeins}
\usepackage{float}
\usepackage{bbm}
\usepackage{subfigure}
\usepackage{verbatim}
\usepackage{amssymb}
\usepackage{amsmath}
\usepackage{amsthm}
\usepackage{mwe}
\usepackage{docmute}
\usepackage{mathtools}
\usepackage[ruled, vlined]{algorithm2e}
\usepackage{bm}
\usepackage{siunitx}
\usepackage{relsize}

\usepackage[disable]{todonotes}

\usepackage{graphicx}

\DeclareMathOperator*{\argmin}{argmin}

\newtheorem{prop}{Proposition}
\newtheorem{definition}{Definition}

\theoremstyle{definition}

\newtheorem{assumption}{Assumption}

\newcommand{\rev}[1]{\textcolor{black}{#1}}

\begin{document}

\title{Mode Selection in Cognitive Radar Networks}
\author{William W. Howard, Samuel R. Shebert, Anthony F. Martone and R. Michael Buehrer
\thanks{W.W. Howard, S.R. Shebert and R.M. Buehrer are with Wireless@VT, Bradley Department of ECE, Virginia Tech, Blacksburg, VA, 24061. \\ 
A.F. Martone is with the U.S. Army Research Laboratory, Adelphi, MD 20783. (e-mail:{anthony.f.martone.civ}@army.mil).\\
Distribution Statement A: Approved for public release. Distribution is unlimited. \\
Portions of this work were presented previously \cite{howard2023_modecontrolconf}. \\
Correspondence:$\{${wwhoward}$\}$@vt.edu  
}
}

\maketitle
\pagenumbering{roman}

\begin{abstract}

Cognitive Radar Networks, which were popularized by Simon Haykin in 2006, have been proposed to address limitations with legacy radar installations. 
These limitations include large physical size, power consumption, fixed operating parameters, and single point vulnerabilities. 
Cognitive radar solves part of this problem through adaptability, using biologically inspired techniques to observe the environment and adjust operation accordingly. 
Cognitive radar \emph{networks} \rev{(\textbf{CRNs})} extend the capabilities of cognitive radar spatially, providing the opportunity to observe targets from multiple angles to mitigate stealth effects; distribute resources over space and in time; obtain better tracking performance; and gain more information from a scene.
Often, problems of cognition in CRNs are viewed through the lens of iterative learning problems - one or multiple cognitive processes are implemented in the network, where each process first observes the environment, then selects operating parameters (from discrete or continuous options) using the history of observations and previous rewards, then repeats the cycle. 
Further, cognitive radar networks often are \rev{modeled} with a flexible architecture and wide-bandwidth front-ends, enabling the addition of electronic support measures such as passive signal estimation. 
In this work we consider questions of the form ``How should a cognitive radar network choose when to observe targets?'' and ``How can a cognitive radar network reduce the amount of energy it uses?''. 
We implement tools from the multi-armed bandit and age of information literature to select \emph{modes} for the network, choosing either an active radar mode or a passive signal estimation mode. 
We show that through the use of target classes, the network can determine how often each target should be observed to optimize tracking performance. 

\end{abstract}


\section{Introduction}
\label{sec:intro}

There is a desire in the literature and within DoD for devices with low size weight and power (\textbf{SWaP}). 
One current method by which this is achieved is through the use of flexible high-speed computing platforms such as FPGAs. 
Modern devices are capable of very fast digital signal processing with very high bandwidths. 
We examine an application of flexible computing in networks of low-power, multi-mode cognitive radar devices (\textbf{CRNs}) \cite{howard2022_MMABjournal} \cite{howard2023_hybridjournal} \cite{haykin2005}. 
These radar devices (\textbf{nodes}) \cite{thornton2022_universaljournal} have the ability to observe targets via active radar observation or passive signal detection and classification (Electronic Support Measures, \textbf{ESM}) and quickly alternating between the two modes, selecting one mode in each of many time steps. 
In addition to the radar nodes, a fusion center (\textbf{FC}) is present to combine target information as well as provide feedback to the radar nodes. 
Over the course of many target tracks, the CRN groups targets by behavior similarity (\textbf{classes}), assuming that the observed targets are generated by a finite number of target distributions. 
We show that if targets of the same class are identifiable by their signal or physical behavior, then more accurate target class estimation will result in 1) a reduction of effective radiated power from each node and 2) an improvement in target tracking performance. 

This ``mode selection'' is motivated by capabilities a CRN may already possess. 
Prior work \cite{9178313} has considered the use of spectrum sensing to determine the utility of a specific radar action in the presence of noise and interference. 
The presence of other devices in a given channel can indicate the quality of radar measurement using that channel. 
Other work \cite{howard2021_multiplayerconf} has considered the use of spectrum sensing to detect and avoid mutual interference, instances of nodes in a CRN interfering with each other.

The targets are not only grouped by motion model similarity, but by signal emission similarity as well. 
We leverage the fact that most modern targets, civilian or military, tend to have characteristic radio emissions (e.g., FM voice communication and Automatic Dependent Surveillance Broadcast (\textbf{ADS-B}) in general aviation; control, telemetry, and data downlink in consumer unmanned aerial vehicles (\textbf{UAVs}); two-way voice communication for hot-air balloons). 
By using passive ESM techniques rather than active radar, the CRN nodes are able to take advantage of additional target information while reducing their power usage and more importantly their radiated power. 
This allows the CRN to perform additional target modeling, associating unique signal characteristics with unique physical characteristics. 
To continue with the previous examples, general aviation aircraft tend to exhibit different kinematics than consumer UAVs, which are in turn physically different from lighter-than-air balloons. 
We show that leveraging these associations can result in the same or better tracking error while requiring less power consumption at the CRN nodes.

The different types of targets are formed into ``classes'' over many independent tracks. 
Then, the CRN nodes can choose a mode of operation that depends on the targets it is currently tracking - choosing radar when it is necessary, and passive ESM when it provides additional information. 
We propose two main categories of decision process: a centralized approach, which considers the entire network and all current tracks for decision-making at the cost of higher communication overhead; and a distributed approach, which allows CRN nodes to select a mode based on the targets they are currently tracking. 

\subsection{Organization}
In Section \ref{sec:background}, we review recent work in the related fields of multi-target tracking, cognitive radar, and radar networks. 
Section \ref{sec:modeling} discusses the model for the CRN considered in this work, as well as specifics on target modeling and tracking. 
Section \ref{sec:methods} develops our centralized and distributed mode selection algorithms. 
In Section \ref{sec:simulations} we present analysis and numerical simulations of our proposed techniques, comparing them against alternatives. 
In Section \ref{sec:conclusions} we discuss our results, draw conclusions, and recommend future work in distributed and centralized control of cognitive radar networks.

\subsection{Contributions} 
We build on previous contributions in the areas of radar network control, cognitive radar, and multi-target tracking. Portions of this work were presented previously \cite{howard2023_modecontrolconf}. We contribute the following to the state of the art: 
\begin{itemize}
    \item A model for mode selection in multi-function cognitive radar networks. 
    \item An analysis of multiple target class formation based on characteristic motion and signal emission models. 
    \item Mathematical analysis of a clustering-based class formation technique. 
    \item A distributed approach which mitigates the effects of network latency. 
    \item A centralized technique which incorporates information from the entire network for decision-making. 
    \item Numerical simulations to support our conclusions. 
    \item We show that our proposed techniques outperform radar-only observation as well as outperforming a random selection algorithm which achieves the same radar observation rate, but does not consider target class formation. 
\end{itemize}

\subsection{Notation}
We use the following notation. 
Matrices and vectors are denoted as bold upper $\mathbf{X}$ or lower $\mathbf{x}$ case letters respectively.
Element-wise multiplication of two matrices or vectors is shown as $\mathbf{X}\odot \mathbf{Y}$. 
Functions are shown as plain letters $F$ or $f$. 
Sets $\mathcal{A}$ are shown as script letters. 
Denote the Lebesgue measure of a set $\mathcal{A}$ as $|\mathcal{A}|$. 
When we wish to show the number of elements in a (finite) set $\mathcal{A}$ rather than its measure, we use the cardinality $\#(\mathcal{A}$). 
The transpose operation is $\mathbf{X}^T$. 
The backslash $\mathcal{A}\backslash \mathcal{B}$ represents the set difference. 
Boxes (intervals) in $\mathbb{R}^d$ are written as $[a,b]^d$ and when the elements of a set are denoted, they are given as $\mathcal{A} = \{a, b, c, \dots \}$. 
Random variables are written as upper-case letters $X$, and their distributions will be specified. 
The set of all real numbers is $\mathbb{R}$ and the set of integers is $\mathbb{Z}$. 
The speed of electromagnetic radiation in a vacuum is given as $c$. 
The Euclidean norm of a vector $\mathbf{x}$ is written as $||\mathbf{x}||$. 
Estimates of a true parameter $p$ are given as $\hat{p}$. 
The operator $x \overset{{\scriptscriptstyle \operatorname{R}}}{\leftarrow} \mathcal{A}$ denotes assigning $x$ as a random sample of the set $\mathcal{A}$.


\section{Background}
\label{sec:background}
CRNs and general cognitive radar are two areas of much recent study. 
Recent contributions in the area of cognitive radar networks include \cite{howard2023_timelyjournal} which investigates Age of Information (\textbf{AoI}) approaches to the problem of timely updating in CRNs, \cite{howard2022_MMABjournal} which investigates cooperative spectrum allocation in CRNs, and \cite{howard2023_hybridjournal} which investigates the relationship and trade-offs between centralized and distributed spectrum allocation techniques. 
Recent contributions in the area of cognitive radar include \cite{thornton2020efficient} which investigates waveform selection, \cite{thornton2022_universaljournal} which discusses a universal learning technique for multi-track optimization, and \cite{Martone_CRN_loop} which discussed ``meta-cognition'', the process of choosing between different cognitive strategies.

\rev{Of particular importance to the current work is \cite{howard2023_timelyjournal}, where the trade-off between centralized cognition and distributed cognition is examined. 
It is shown that centralized systems have some inherent problems, such as the need to move information through the network to a central location and the delays that this movement can cause. 
In many decision-making problems, timeliness is very important and this network latency can impact the performance of centralized techniques. }
There has also been recent work in the area of network-based passive target localization. 
The authors of \cite{6807568} present a centralized passive estimation network, where the nodes act as amplify-and-forward units, and the total amplification of the network is power-limited. 
All of the decision-making in the network is located in the fusion center, and is confined to allocating the limited amplification power to the nodes. 
The targets are modeled as transmitting one of several complex-valued signals. 
The FC fuses the node measurements with a goal of estimating the true target signal. 
This work is useful because it provides a framework for signal estimation, which could support classification. 
The goal of \cite{6807568} is simply to reconstruct (with high accuracy) the signals emitted by the target. 
Our work focuses (among other goals) on a slightly different aspect of signal analysis - we're concerned with identifying the \emph{type} of signal emitted by a target, rather than reconstructing it. 

\rev{The problem of multi-sensor target identification has been previously studied in \cite{Bogler1987, Hong1993} and more recently in \cite{Lei2020, Li2022}. 
These works use Dempster-Shafer evidence theory to fuse measurements from multiple sensors to improve target identification rates. 
In \cite{Challa2001} and \cite{Cao2018}, radar and passive signal classification are combined to improve tracking performance by classifying targets. 
However, it is assumed that the pairing between radar targets and emitted signals are known \textit{a priori}, which may not be known in practice.
Another use of multi-function radar networks is in dual radar-communication systems \cite{8386661}, \cite{8114253} where communications information is embedded into transmitted radar waveforms. 
The particular problem of radar mode control (i.e., selecting one of several operational modes in each time step for a radar system) has been investigated before \cite{10038921}
In addition, the problem of fusing information from discrete radar and ESM sensors has been investigated \cite{rs15163977}
However, we believe that the particular problem of cognitive mode selection in radar networks has not yet been investigated. 
}

\rev{
One of the main problems with passive and active sensing in such radar systems is \emph{sensor fusion} \cite{8943388}, \cite{7979175}. 
Sensor fusion is a rich field of literature containing several seminal contributions, such as the development of Kalman tracking filters \cite{1271397}
Tracking filters use estimates of measurement noise and process noise combined with actual sensor measurements to provide an estimate of the target location. }

Multi-target tracking is also a rich field, containing several important contributions. 
In particular, the study of \emph{probability hypothesis density} (\textbf{PHD}) filters \cite{1710358} \cite{5259179} aims to solve the problem of \emph{target association}. 
\rev{PHD filters are very useful in scenarios with multiple targets and multiple sensors. }
The PHD is the first moment of the target state space, and is akin to the expected value of a random variable. 
It allows for the estimation of the number of targets in a scene as well as the propagation of the state of those targets through time. 
As target detections are received (from point or extended object models, among others), the PHD filter is updated. 
In \cite{4516991} a multiple model PHD filter is developed, where several PHD filters are used in parallel to estimate the motion model of targets.


\section{Target Modeling}
\label{sec:modeling}
\subsection{Discrete Time Markov Chains}
\label{ss:dtmc}
We consider targets in a region $B$ which are spatially distributed according to a stochastic Poisson point process (\textbf{PPP})
A Poisson number of targets with mean $\lambda_M$ is drawn, and each target is assigned an initial state uniformly at random in the region $B$. 
Generally, each target in the observable region $B$ can be partially or wholly described by several parameters. 
Let the parameters describing a target $m$ be collected into a set $\mathcal{X}_m$. 
In other words, a given element $E_m\in\mathcal{X}_m$ describes some quality of the target $m$. 
We limit our consideration of target $m$ to those parameters which can be described as time-homogeneous ergodic Markov chains in a finite state space. 
So, a parameter $E_m$ is a time-varying quantity that can take on finitely many values. 
Since $E_m$ is a discrete-time Markov chain (\textbf{DTMC}), it has the Markov property given in Eq. (\ref{eq:markov}). 
This means that the probability of transitioning to any given state is only dependent on the current state and not the history; i.e. there is no memory length. 
\begin{align}
\label{eq:markov}
    \begin{split}
        \Pr[E_m^{(t+1)}] = \Pr[e_m &\;|\; E_m^{(1)} = e_m^{(1)}, E_m^{(2)} = e_m^{(2)}, \dots, \\
        E_m^{(t)} = e_m^{(t)}]\\
        &= \Pr[E_m^{(t+1)} = e_m | E_m^{(t)} = e_m^{(t)}]
    \end{split}
\end{align}

We say that there are $N_E$ states which form a finite set $\mathcal{E}$ called the state space of the DTMC. 
The transition matrix $P_m^E$ consists of entries 
\begin{equation}
    \label{eq:markov_transition}
    p_{ij} = \Pr[E_m^{(t+1)} = j | E_m^{(t)} = i]
\end{equation}
which describe the probability of transitioning states at any time step $t$. 
Note that all rows sum to $1$ and all elements are non-negative. 
Since the process is time-homogeneous, the transition matrix is not time dependent. 
Lastly, the stationary distribution of parameter $E_m$ can be defined as Eq. (\ref{eq:stationary}). 
The stationary distribution $\pi_m^E$\footnote{In the context of DTMCs, we slightly abuse notation to allow $\pi_m^X$ to denote the parameter described by $\pi$, rather than exponentiation. } can be seen to be a normalized left eigenvector of the transition matrix \cite{markov_model_book}. 
\begin{equation}
    \label{eq:stationary}
    \pi_m^E = \pi_m^E P_m^E
\end{equation}
Since the process is ergodic, the stationary distribution is also the limiting distribution of any starting distribution \cite{norris1998markov}. 
In other words, there is only one eigenvector of $P_m^E$ with an eigenvalue of 1. 
This means that as the number of samples $n$ trends towards infinity, the empirical stationary distribution is not dependent on the initial state distribution. 

\subsection{Class Definitions}
\label{ss:definitions}
\begin{definition}[Equal in State Distribution]
    \label{def:eq_in_dist}
    Two random variables $X$ and $Y$ are said to be \emph{equal in state distribution} if the following properties hold: 
    \begin{enumerate}
        \item Both $X$ and $Y$ can be described as DTMCs. 
        \item If $N_X$ is the size of the state space for $X$ and $N_Y$ is similarly defined for $Y$, then $N_X = N_Y$. 
        \item If $\pi^X$ is the stationary distribution for $X$ and similarly for $Y$, then $\pi^X = \pi^Y$. 
    \end{enumerate}    
\end{definition}
%
%
\begin{definition}[Target Class]
    Let the parameters describing target $m_0$ be collected into $\mathcal{X}_{m_0}$. 
    If there exists a target $m_1$ with the following property, then it is said to be of the same \emph{class} as $m_0$. 
    Denote the class as $C$. 
    \begin{itemize}
        \item Each element of $\mathcal{X}_{m_0}$ maps to a unique element of $\mathcal{X}_{m_1}$. Corresponding elements are defined over the same state space and are equal in state distribution. 
    \end{itemize}
\end{definition}
For a given target class $C$ with a given parameter $E$, we say that the set $\mathcal{C}$ contains the state space $(\mathcal{E})$ of $E$  as well as the stationary distribution $\pi^E$ over $\mathcal{E}$ which defines the class. 
\begin{definition}[Target Family]
    \label{def:family}
    A \emph{target family} is a group of target classes $\{C_1, C_2, \dots\}$ with the following properties. Call the family $F$. 
    \begin{enumerate}
        \item If state space $\mathcal{E}$ is contained in $\mathcal{C}_1$, then it is also contained in $\mathcal{C}_i$ for all $\mathcal{C}_i \in \{C_1, C_2, \dots\}$. 
        \item If a stationary distribution $\pi^E$ is contained in $\mathcal{C}_i$, then it is \emph{not} in $\mathcal{C}_j$ for any $\mathcal{C}_j \in \{C_1, C_2, \dots\}\backslash C_i$. 
    \end{enumerate}
\end{definition}
In other words, all classes within a family are defined by the same set of parameters, although they do not have the same parameter values. 
We discuss in Section \ref{sec:estimation} how a CRN node forms an estimate of target parameters.

\begin{prop}[Unique Class]
    \label{prop:unique}
    Let node $n$ draw an estimate $\hat{\pi_m^E}$ of the stationary distribution of parameter $E_m$ for target $m$. 
    The node estimates the corresponding class, $\hat{C}$. 
    Further, let target $m$ be drawn from class $C$ in a family $F$. \\
    If $\hat{\pi_m^E} = \pi_m^E$, then $\hat{C} = C$. 
\end{prop}

\begin{proof}[Proof of Prop. \ref{prop:unique}]
    Trivial by Def. \ref{def:family}
\end{proof}

So, targets $m$ within a family $F$ have the useful property that if one parameter's stationary distribution $\pi_m^E$ can be estimated, then the class can also be estimated. 

\subsection{Motion Modeling}
As target $m$ moves through space, its state at a time $t$ can be described by Eq. (\ref{eq:state}) which is comprised of position and velocity. 
We summarize the higher-order derivatives by defining \emph{motion models} \cite{ristic2003beyond}, which describe how the target's state propagates in time. 
Let $V_m(t)$ be the motion model for target $m$, and say that it takes on one of $N_V$ states for each time step. 
In other words, $V_m(t)$ is a time-homogeneous ergodic discrete-time Markov chain in a finite state space. 
Let $\pi_m^V$ be the stationary distribution of the states in $\mathcal{V}$, the space of possible motion models. 
Lastly let $P_m^V$ be the motion model state transition probability matrix for target $m$. 
\begin{equation}
    \label{eq:state}
    \mathbf{X}_m(t) = [x_m(t), \dot{x}_m(t), y_m(t), \dot{y}_m(t), z_m(t), \dot{z}_m(t)]
\end{equation}
\rev{Note that $V_m(t)$ describes the \emph{motion model state} while $\mathbf{X}_m(t)$ describes only the current position and velocity of the target. 
In order to describe the higher-order derivatives of position we use $V_m(t)$. }
In Section \ref{ss:radar}, we discuss how the motion model for target $m$ is estimated. 

\subsection{Signal Modeling}
\label{ss:sig_model}
Each target $m$ also operates radio equipment (i.e., radar, communications, telemetry, etc. ) resulting in electromagnetic emissions that can be detected via passive ESM. 
We say that $S_m(t)$ is the state of the emissions from target $m$ at time $t$, and that $S_m(t)$ takes on one of the $N_S$ values in the space of possible emissions $\mathcal{S}$. 
So, $S_m(t)$ is a time-homogeneous ergodic discrete-time Markov chain in a finite state space. 
Let $\pi_m^S$ be the stationary distribution for target $m$, and say $P_m^S$ is the emission state transition probability matrix. 
We consider the absence of a signal to be a signal state. 
So, we say that one of the states in the space consists of no signals. 
In Section \ref{ss:esm} we define the \emph{probability of interception} of target signals. 

\subsection{Summary}
We model a target $m$ using two Markov processes. 
The set $\mathcal{X}_m$ contains the motion model $V_m(t)$ which takes on one of $N_V$ states with transition probabilities $P_m^V$ and stationary distribution $\pi_m^V$. 
Further, it contains $S_m(t)$, the emission state which takes on one of $N_S$ states with transition probabilities $P_m^S$ and stationary distribution $\pi_m^S$. 
Both processes are ergodic time-homogeneous discrete-time Markov chains in finite state spaces, which means that the transition probabilities are not time dependent and there is a nonzero probability of transitioning from any state $i$ to any state $j$. 

There exist \emph{classes} of targets which are composed of targets whose motion and emission states are equally distributed\footnote{Note that this does \emph{not} imply that the motion or emission state of two targets in the same class are identical, just that they are identically distributed over the respective state spaces. }. 
Finally, the target classes form a \emph{family} which provides uniqueness to each class.

\section{Target Estimation}
\label{sec:estimation}
We consider a network which is capable of performing either active radar observation or passive signal parameter estimation (ESM). 
\rev{Examples of such a network might include networked primary air surveillance radar or vehicular radar in connected vehicles. 
Connected vehicle radar typically operates in the Industrial, Scientific, and Medical (\textbf{ISM}) band at $\qty{24}{\giga\hertz}$. 
An ESM receiver could operate in this same band and detect radar pulses from similar devices in other networks. }

Due to target class estimation, ESM estimation provides the ability to classify a target and reduce subsequent tracking error. 
In addition, selecting the passive ESM action lowers the observabilty of the radar network to adversaries. 
This is quantified by the \emph{maximum intercept range} $R_{I, \text{max}}$, Eq. (\ref{eq:radar_max_intercept}), where $P_{\text{radar}}$ is the transmit power, $G_t$ is the transmit gain in the direction of the intercept receiver, $G_I$ is the intercept receiver gain, $L_1$ is one-way atmospheric loss, $\lambda$ is the wavelength of the center frequency, $L$ is system loss, $\delta_I$ is the intercept receiver sensitivity, $F_1$ is the intercept receiver noise figure, $B_1$ is the intercept receiver bandwidth, SNR$_{Ii}$ is the intercept receiver SNR, $k$ is Boltzmann's constant, and $T_0$ is the intercept receiver noise temperature. 
\begin{align}
\label{eq:radar_max_intercept}
    R_{I,\text{max}} &= \sqrt{\frac{P_{\text{radar}}G_tG_IL_1\lambda^2}{(4\pi)^2L\delta_I}}\\
    \delta_I &= kT_0F_1B_1(SNR_{Ii})
\end{align}
This is the maximum range at which an interceptor can detect the radar. 
When the number of radar transmissions is reduced, this reduces the average distance at which a transmission may be intercepted. 
The max intercept range for a single radar node with a variable transmit probability (i.e., portion of time spent selecting active radar) is shown in Fig. \ref{fig:radar_max_intercept}. 

\begin{figure}
    \centering
    \includegraphics[scale=0.65]{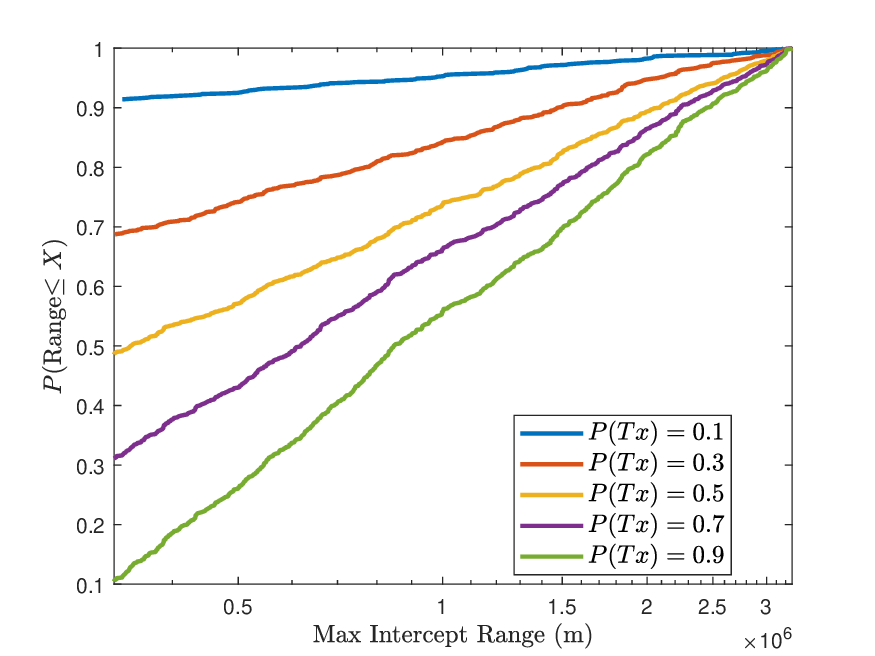}
    \caption{Maximum intercept range for a single radar node with a variable transmit probability. }
    \label{fig:radar_max_intercept}
\end{figure}

\subsection{Network Structure}
Following a similar structure to those presented in \cite{howard2023_timelyjournal} \cite{howard2023_timelyconf}, we use tools from stochastic geometry \cite{haenggi} and consider a compact region of $\mathbb{R}^3$, assuming that the network consists of a random set $\mathcal{N}$ of multi-function radar/ESM nodes generated by a Poisson Point Process \cite{8464057} with density $\lambda_N$ and positions $\textbf{X}_n$. 
Each node $n$ ``covers'' a region $C_n$ with measure $|C_n|$, which accounts for practical range limit for radar devices. 
Let $N$ with mean $\overline{N}$ be the random variable describing the number of such nodes. 
The distribution of $N$ is given as Eq. (\ref{eq:poisson_dist}), with the mean shown in  Eq. (\ref{eq:poisson_mean}). 
\begin{equation}
    \label{eq:poisson_dist}
    \Pr[N=n] = \frac{\lambda_N |B|^n e^{-\lambda_N |B|}}{n!}
\end{equation}
\begin{equation}
    \label{eq:poisson_mean}
    \overline{N} = \lambda_N|B|
\end{equation}
Similarly, there is a random set $\mathcal{M}$ containing $M$ targets at positions $\textbf{X}_m$, generated by a PPP with density $\lambda_M$ and mean $\overline{M} = \lambda_m|B|$. 
When we consider numerical simulations, a single simulation consists of a single realization of this model. 

Targets continue to exist in the next time step according to a fixed probability, and a Poisson number of new targets are born in each time step \cite{6507656}. 
These two events cancel out such that the target density is maintained. 

The coverage of the network is the union of the covered region for each node, Eq. (\ref{eq:coverage}). 
\begin{equation}
    \label{eq:coverage}
    \mathbf{C} = \bigcup_{n\in\mathcal{N}}C_n
\end{equation}
As a result, the probability that a target is ``covered'' is given by Eq. (\ref{eq:target_covered}). 
This is also the probability that any random point in $B$ is covered. 
\begin{equation}
    \label{eq:target_covered}
    \Pr[\mathbf{X}_m\in\mathbf{C}] = 1 - e^{-\lambda_N\mathbb{E}[|C_n|]}
\end{equation}
An example of this network geometry is shown in Fig. \ref{fig:stick_system}.

\begin{figure}
    \centering
    \includegraphics[scale=0.65]{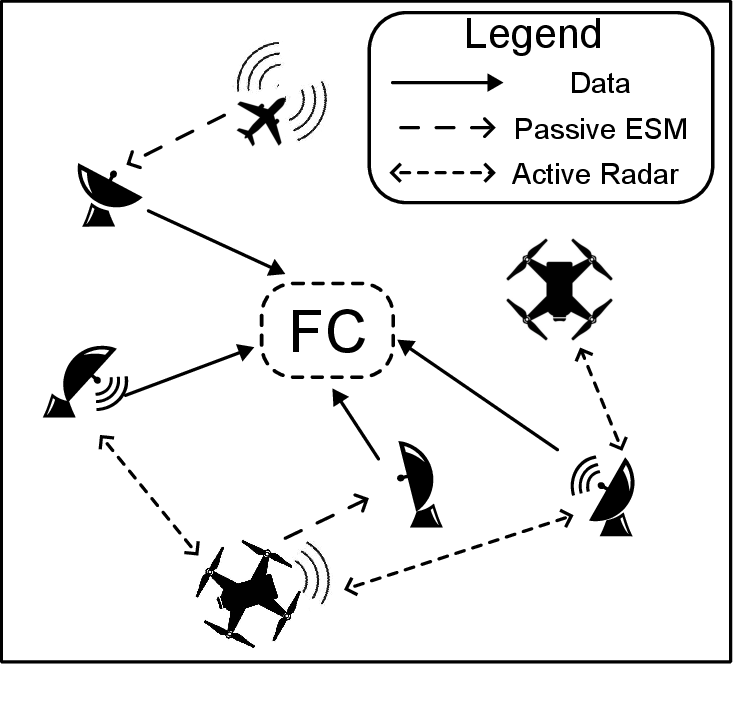}
    \caption{A model of the type of network we consider, where each node can choose between active and passive observation of several types of targets. The nodes send observations to the fusion center. }
    \label{fig:stick_system}
\end{figure}

\subsection{Active Radar}
\label{ss:radar}
Radar estimation allows for remote target position and velocity observation. 
We consider a monostatic radar configuration, where each node is capable of transmitting a train of coherent pulses, receiving scatter from clutter and target responses, then coherently combining the pulses and performing range-Doppler estimation. 
Constant false alarm rate (\textbf{CFAR}) detection is performed, by which the radar node estimates the presence or absence of targets in the range-Doppler map. 
We assume that when the radar mode is used, the radar uses a multi-function detection and tracking mode. 
In other words, it scans for new targets and attempts to obtain further measurements of previous tracks. 

    

\subsubsection*{Measurement}
\label{sss:radar_measurement}
Following the extended object model \cite{granstrom2017extended}, we assume that each target visible to the radar may occupy several resolution cells and therefore generate multiple detections per time step. 
Let node $n$ cover $\mathcal{M}_n^{(t)}$ targets at time $t$. 
Since the targets can move around, spawn, or retire \cite{6178085}, this quantity is time-varying. 
Note that 
\begin{equation}
    \label{eq:expected_target_coverage}
    \mathbb{E}[\#\mathcal{M}_n^{(t)}] = \lambda_M|S_n|
\end{equation}
is the expected number of targets covered by any given node. 
So, each target $m$ in the observable region for node $n$ generates a number of detections $\rev{\#(Z_{mn})}$: 
\begin{equation}
    \mathbf{Z}_{mn} = \left\{ \mathbf{z}^{(j)}\right\}_{j=1}^{\rev{\#(Z_{mn})}}
\end{equation}
Each target detection is missed with probability $\rev{1-}P_D$. 
\rev{As discussed in Sec. \ref{ss:tracking}, the sets represented by $\mathbf{Z}$ contain the radar observations. }
Since more than one target may be in the region covered by node $n$, the total number of target detections generated at time $t$ by node $n$ is given as Eq. (\ref{eq:radar_detections}). 
\begin{equation}
    \label{eq:radar_detections}
    \mathbf{Z}_n^t = \bigcup_{m\in\mathcal{M}_n^{(t)}}\mathbf{Z}_{mn}
\end{equation}
Then, a PHD filter \cite{1710358} \cite{4516991} is used to associate \cite{5259179} these target detections with previous tracks and generate new tracks if necessary. 
Since $\mathbf{Z}_n$ may contain false alarms generated at a rate of $\lambda_{FA}$, the PHD filter maintains a list of tentative and confirmed tracks. 
False alarms are uniformly distributed in the region. 

At a time step $t$, the output from the PHD filter is used as the current state of the estimated target track. 
Let $\hat{\mathcal{M}}_n^{(t)}$ be the set of active tracks at node $n$. 
Then, for a track $\hat{m}\in \hat{\mathcal{M}}_n^{(t)}$, $\mathbf{X}_{\hat{m}}$ represents the current state. 
Further tracking details are provided in Section \ref{ss:tracking}.

\subsection{Passive Electronic Support Measures}
\label{ss:esm}
Passive ESM estimation allows the node to make observations of targets without the need for high-power radio emissions. 
Unlike active radar sensing, ESM performance is dependent on the characteristics of the transmitting equipment. 
The primary values which impact the performance of the ESM receiver are the \emph{maximum detectable range} and the \emph{probability of intercept}. 

\subsubsection{Maximum Detectable Range}
\label{sss:max_range}
The maximum detectable range of targets\footnote{Note that this is the maximum range at which a node in the CRN can passively detect a target, which is different than Eq. (\ref{eq:radar_max_intercept}\rev{)} which is the maximum range at which an intercept receiver could detect a node in the CRN. } is dependent on the received SNR\footnote{Specifically, this is the instantaneous or non-integrated SNR. }, Eq. (\ref{eq:ESM_SNR}). 
\begin{equation}
    \label{eq:ESM_SNR}
    \text{SNR}^{ESM}_{mn} = \frac{P_t G_t G_r \lambda^2}{(4 \pi R)^2 P_n L}
\end{equation}
Note that the SNR from a target $m$ received at the $n^{th}$ node is dependent on several aspects of the target equipment: the transmit power of the target equipment $(P_t)$, the transmit antenna gain $(G_t)$, and the wavelength $(\lambda)$. 
In addition, there are a few other parameters which depend on the $n^{th}$ receiving ESM: the receive antenna gain $(G_r)$, the range to the target $(R_{mn})$, the receiver noise power $(P_n)$ and any other losses $(L)$. 
The receiver noise power is given as Eq. (\ref{eq:rx_noise}), where $T_0$ is $\qty{290}{\kelvin}$, $F$ is the receiver noise figure $(10\text{ dB})$, $B$ is the receiver bandwidth $(\qty{1}{\mega\hertz})$, and $k$ is Boltzmann's constant. 
\begin{equation}
    \label{eq:rx_noise}
    P_n = k T_0 F B
\end{equation}
Fig.  \ref{fig:ESM_SNR} \rev{compares the received instantaneous SNR for targets over range and transmit power.}
Of course, the SNR required for high probability detection with a low probability of false alarm will vary based on the detection algorithm, the signal capture duration, and the signal of interest. 
For cyclostationary detectors with a signal duration greater than a few milliseconds, it is possible to detect signals at a rate approaching $100\%$ at or below $\qty{0}{dB}$ SNR with a false alarm rate less than $1/100$ \cite{6914549} \cite{5757489}. 
Thus, we make Assumption \ref{ass:ESM_detection} below. 
\begin{figure}
    \centering
    \includegraphics[scale=.65]{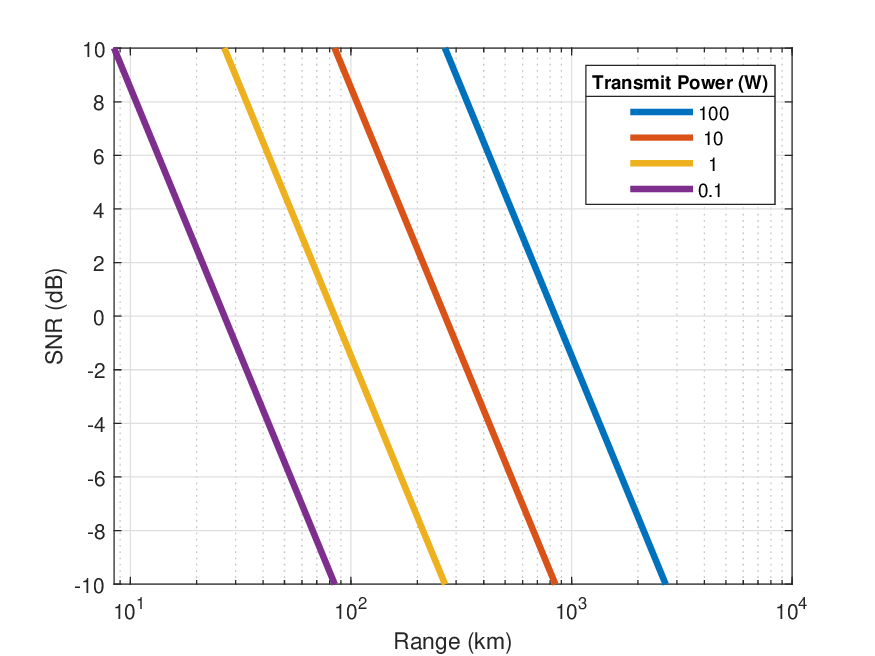}
    \caption{ESM receiver \rev{instantaneous} SNR for a center frequency of $\qty{1}{\giga\hertz}$, omnidirectional receive antennas (gain of \rev{$\qty{0}{dB}$}), and $\rev{\qty{3}{dB}}$ of losses. Given an SNR requirement of $\qty{0}{dB}$ for detection \rev{\cite{6914549, 5757489}}, targets can be detected at ranges of up to $\qty{100}{\kilo\meter}$, depending on the transmitter power. }
    \label{fig:ESM_SNR}
\end{figure}
\begin{definition}[Maximum Detectable ESM Range]
    The maximum detectable range $R_{mn}$ for electronic support measures between node $n$ and target $m$ is the maximum range for which Eq. (\ref{eq:max_range}) holds, where the $\text{SNR}$ is given by Eq. (\ref{eq:ESM_SNR}). 
    \begin{equation}
        \label{eq:max_range}
        \text{SNR}_{mn} \geq 0
    \end{equation}
    Denote $R_{mn} \text{ s.t. } \text{SNR}_{mn} \geq 0$ as $R_{mn}^{\text{ESM}}$
\end{definition}
Note that the \emph{probability of correct detection} $P_{d, mn}^{ESM}$ is the probability that a signal is correctly identified, and the \emph{probability of false alarm} $P_{fa, mn}^{ESM}$ is the probability that a detected signal is misidentified or mistakenly detected (i.e. not present). 
\begin{assumption}[In-Range Targets are Detectable]
    \label{ass:ESM_detection}
    Targets for which $R_{mn} < R_{mn}^{ESM}$ have $P_{d, mn}^{ESM}=1$ and $P_{fa, mn}^{ESM}$ = 0. 
\end{assumption}
While detectors with high $P_d$ and low $P_{fa}$ exist in the literature, the signal must still be present. 
We account for this by also defining the probability of intercept. 

\begin{figure}
    \centering
    \includegraphics[scale=.65]{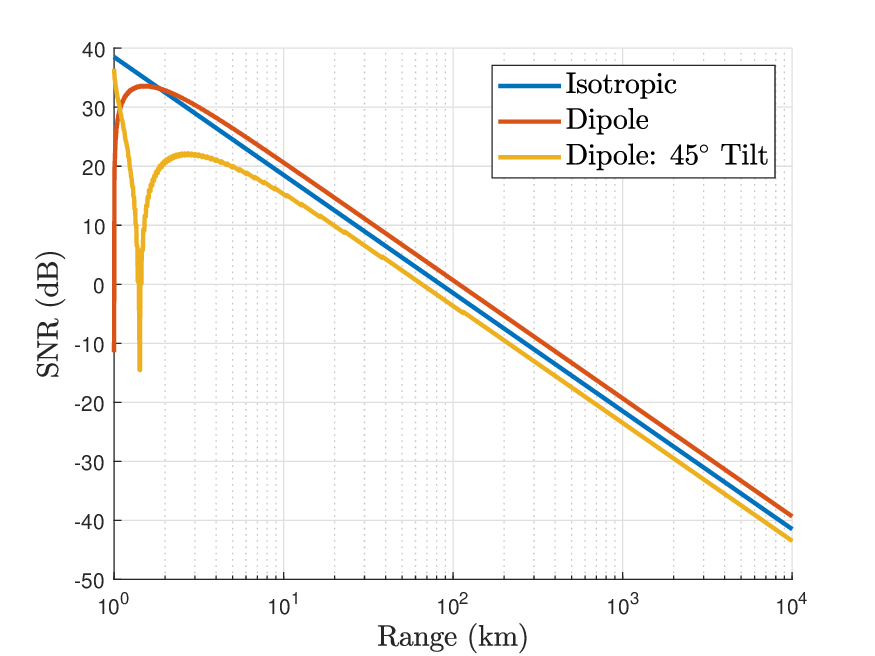}
    \caption{\rev{ESM receiver instantaneous SNR for multiple transmit antennas and orientations. The same system parameters as Figure \ref{fig:ESM_SNR}, with a 1 Watt transmit power.}}
    \label{fig:ESM_antenna}
\end{figure}

\subsubsection{Probability of Intercept}
\rev{The probability of intercept $P_{I,mn}^{ESM}$ is a term used to characterize the uncertainties in look angle, terrain, transmission rate, etc, which decrease the rate at which an in-range target is detected\footnote{\rev{Note that while $P_{d,mn}^{ESM}$ may be high, a low probability of intercept will still reduce the probability that a target signal can be identified. }} \cite[Ch. 4]{9101176}. In this work, the probability of intercept is modeled as:
\begin{equation}
    P_{I,mn}^{ESM} = \text{Pr}[I_{mn}^{S} \cap I_{mn}^{F} \cap I_{mn}^{T}]
\end{equation}
$I_{mn}^{S}$, $I_{mn}^{F}$, and $I_{mn}^{T}$ are the events where ESM node $n$ intercepts target $m$'s emission in space, frequency and time, respectively. The intersection of all events (in addition to the target being \textit{in-range}) is required for the emission to be identified. These events are assumed to be independent, such that:
\begin{equation}
    P_{I,mn}^{ESM} = P_{I,mn}^{S} P_{I,mn}^{F} P_{I,mn}^{T}
\end{equation}
where $P_{I,mn}^{S} = \text{Pr}[I_{mn}^{S}]$, and accordingly, $P_{I,mn}^{F}$ and $P_{I,mn}^{T}$ are the probabilities for frequency and time. These probabilities are the average rates that the ESM node's measurement intersects with the target's emission in the respective domains.}


\rev{For example, in the spatial domain, $P_{I,mn}^{S}$ is the average rate that the target's antenna illuminates the ESM receiver.
Figure \ref{fig:ESM_antenna} depicts this through the instantaneous SNR at an ESM receiver. In the case of the perfectly isotropic antenna, $P_{I,mn}^{S} = 1$, as the antenna radiates in all directions with equal gain. For practical antennas, such as a dipole, $P_{I,mn}^{S} < 1$. At particular antenna orientations and ranges, a null in the antenna pattern will be oriented towards the ESM receiver. This causes the SNR to drop below the detectable range and the ESM receiver to fail to intercept the target's emissions. These 'SNR outages' are shown in Figure \ref{fig:ESM_antenna} for the dipole antenna between 1 and 2 km ranges. Note that the SNR outage occurs at different ranges for the two dipole antenna orientations considered. This is because the geometry that causes the antenna null to point towards the ESM receiver jointly depends on orientation and range.}


\rev{Since the targets are maneuvering, the orientation of the targets antenna and range to the ESM receiver will change over time. Therefore, the SNR outages are treated as randomly occurring failed interceptions. Figure \ref{fig:ESM_antenna} considers dipole antennas, but this model applies to more directional antennas as well. In general, as the target's antenna becomes more directional, the maximum gain of the antenna in a particular direction increases (i.e. the maximum detectable range increases) at the expense of gain in other directions (i.e. the probability of intercept decreases).}

\rev{$P_{I,mn}^{T}$ and $P_{I,mn}^{F}$ are the probabilities that the ESM receiver's measurement in time and frequency overlap with the target's emission in time and frequency, respectively. Unlike the spatial domain, realistic cases exist where $P_{I,mn}^{T}$ or $P_{I,mn}^{F}$ would be equal to 1. For example, if the target is always transmitting, $P_{I,mn}^{T} = 1$, and if the ESM receiver's bandwidth includes all frequencies used by the transmitter, $P_{I,mn}^{F}=1$.}

\subsection{Tracking Formulation and Fusion}
\label{ss:tracking}
In each time step $t$, every node selects one of the two available actions (active radar or passive signal classification) and receives a vector of detections. 
These detections are then filtered and associated to target tracks. 
We follow the development of \cite{4516991} and use the multiple model PHD filter to estimate state of each target, including motion model. 
From Eq. \ref{eq:radar_detections} we have
\begin{equation*}
    \mathbf{Z}_n^t = \bigcup_{m\in\mathcal{M}_n^{(t)}}\mathbf{Z}_{mn}
\end{equation*}
Let $\mathbf{Z}_n^{1:t}$ denote the ordered set of observations until time $t$. 

From \cite{4516991}, we have the multiple model PHD tracking filter for maneuvering targets. 
We reproduce the core steps of the technique here. 
The filter begins with an initial density Eq. (\ref{eq:initial_density}). 
\begin{equation}
    \label{eq:initial_density}
    \tilde{D}_{t|t-1}\left(\mathbf{X}(t-1), V(t)=i | \mathbf{Z}_n^{1:t-1}\right)
\end{equation}

The motion model mixing is given by Eq. (\ref{eq:PHD_mixing}), where $V(t)$ is one of several motion model states, $N_V$ is the number of motion model states, and $P_{ij}$ are the transition matrix entries, Eq. (\ref{eq:markov_transition}), for transitioning from state $i$ to state $j$. 
Note that there is a separate PHD for each possible motion model. 
\begin{multline}
    \label{eq:PHD_mixing}
    \tilde{D}_{t|t-1}\left(\mathbf{X}(t), V(t)=i | \mathbf{Z}_n^{1:t-1}\right) = \\
    \sum_{j=1}^{N_V} D_{t-1|t-1}\left(X(t-1), V(t-1)=j | \mathbf{Z}_n^{1:t-1}\right) P_{ij}, \\
    i = 1:N_V
\end{multline}
The PHD prediction step is given as Eq. (\ref{eq:PHD_predict}), where $\gamma_t(\cdot)$ is the target birth PHD, $e_{t|t-1}(\cdot)$ represents the probability that each target survives to the next round, $f_{t|t-1}$ and is the motion state conditioned target likelihood.  
\begin{multline}
    \label{eq:PHD_predict}
    D_{t|t-1}\left(X(t), V(t)=i | \mathbf{Z}_n^{1:t-1}\right) = \\
    \gamma_t(X(t), V(t)=i) + \\
    \int \left[ e_{t|t-1}\left(X(t)\right)f_{t|t-1}\left(X(t)|X(t-1), V(t)=i\right)\right]\times \\
    \tilde{D}_{t|t-1}\left(X(t), V(t)=i | \mathbf{Z}_n^{1:t-1}\right) dX(t-1)
\end{multline}
Finally, the PHD update step is given as Eq. (\ref{eq:PHD_update}), where $P_D$ is the probability of detection, $\lambda_{FA}$ is the false alarm rate, $C_{FA}$ is the false alarm spatial distribution, and $\Psi_t(\cdot)$ is the PHD likelihood function, Eq. (\ref{eq:PHD_likelihood}). 
\begin{multline}
    \label{eq:PHD_update}
    D_{t|t}\left(X(t), V(t)=i)|\mathbf{Z}_n^{1:t}\right) \cong \\
    \left[\sum_{\mathbf{z}\in\mathbf{Z}_n^t} \frac{P_D(X(t)) f_{t|t}\left(\mathbf{z}|X(t), V(t)=i\right)}{\lambda_{FA}C_{FA} + \Psi_t\left(\mathbf{z}|\mathbf{Z}_n^{1:t-1}\right)} + \left(1-P_D(X(t))\right)\right] \times \\
    D_{t|t-1}\left(X(t), V(t)=i|\mathbf{Z}_n^{1:t-1}\right)
\end{multline}
\begin{multline}
    \label{eq:PHD_likelihood}
    \Psi_t\left(\mathbf{z}|\mathbf{Z}_n^{1:t-1}\right) = \\
    \int\big[ P_D(X(t)) f_{t|t}\left(\mathbf{z}|X(t), V(t)=i\right) \times \\
    D_{t|t-1}\left(X(t), V(t)=i|\mathbf{Z}_n^{1:t-1}\right) \big] dX(t)
\end{multline}
The motion model at time $t$ for target $m$ is estimated by integrating the PHD filters in the vicinity of target $m$ and taking the maximum. 
When class motion model probabilities are substituted for estimated motion model probabilities, this is done via $P_{ij}$. 
\rev{Estimating the motion model parameters results in lower tracking error as time goes on, as shown in Fig. \ref{fig:tuned_filters}. }

\rev{The network exhibits latency between the times at which nodes detect targets, and the times at which node observations arrive at the CC. 
As will be discussed later, this effect is particularly of interest for algorithms which combine data centrally for decision-making. } 
Let $\tau_{m,n}^V$ denote the times when node $n$ detects target $m$ using active radar, and similarly let $\tau_{m,n}^S$ denote the times when node $n$ detects target $m$ using passive signal classification. 
\rev{Then, if the network exhibits a latency of $\sigma_L$, the CC will receive that node's update at time $\tau + t_L$ where $t_L$ is a log-normal random variable \cite{ANTONIOU200272} \cite{8691027} with parameters $\mu=0$ and $\sigma_L$. 
The PDF of a log-normal random variable is given as Eq. \ref{eq:lognormal}. 
\begin{equation}
    t_L \sim \text{Lognormal}(0, \sigma_L^2)
\end{equation}
\begin{equation}
    \label{eq:lognormal}
    f_T(t_L) = \frac{1}{t_L\sigma_L\sqrt{2\pi}} \exp{\left(-\frac{(\ln(t_l-\mu)^2}{2\sigma_L^2}\right)}
\end{equation}
As the network latency increases, we should expect to see tracking performance decrease. 
This is because even though the individual node tracking filters are not impacted by latency, the delayed movement of information through the network results in decreased decision-making performance and poor mode selection. 
}

\subsection{Class Formation}
\begin{figure}
    \centering
    \includegraphics[scale=0.65]{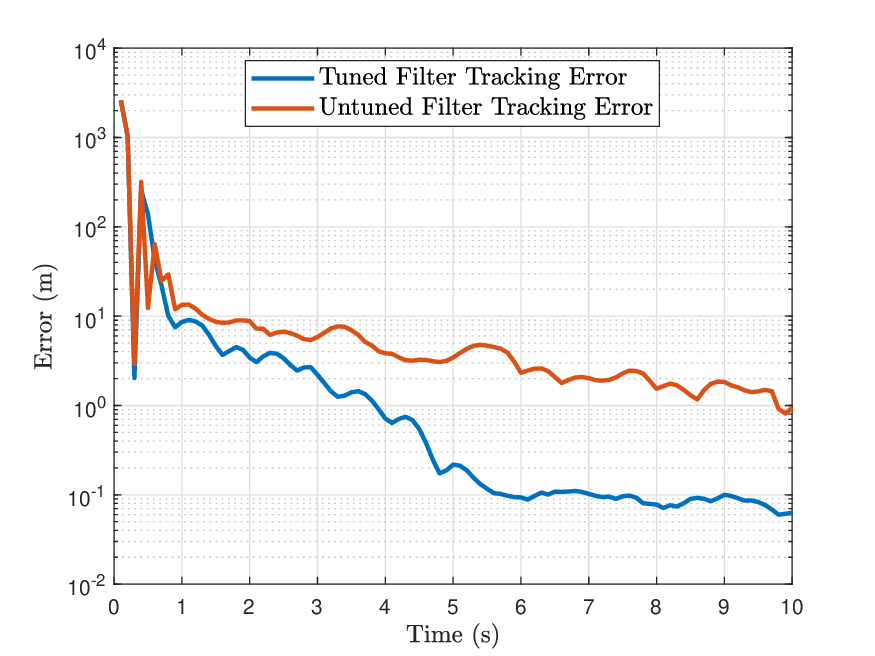}
    \caption{Kalman filters which are tuned to the process noise and motion model probabilities for the class will provide lower tracking error than equivalent filters which are ``untuned''. }
    \label{fig:tuned_filters}
\end{figure}
Since each node can estimate the motion model $\hat{V}_m(t)$ and signal state $\hat{S}_m(t)$ for each target $m\in\mathcal{M}_n^{(t)}$ (the set of targets observable by node $n$ at time $t$), given multiple observations the node can also form estimates of the Markov processes which generate the motion model and signal states. 
Recall that $\tau_{m,n}^V$ denotes the set of times when node $n$ detected target $m$ using active radar, and $\tau_{m,n}^S$ are similarly the times when node $n$ detected target $m$ using passive signal classification. 
Given $t_1$ active and $t_2$ passive observations of target $m$, the motion model stationary distribution $\pi_m^V$ and the signal model stationary distribution $\pi_m^S$ can be estimated as Eq. (\ref{eq:est_motion_stationary}) and Eq. (\ref{eq:est_signal_stationary}). 
Note that this holds due to the ergodicity of the Markov chains: the stationary distribution is the limiting distribution for any initial condition. 
\begin{align}
\label{eq:est_motion_stationary}
    \hat{\pi}_{m,n}^V &= \left[\frac{\sum_{t\in\tau_{m,n}^V}I_{m,n}^V(t, v)}{\#\tau_{m,n}^V} \text{ for $v \in \mathcal{V}$}\right]\\
    I_{m,n}^V(t, v) &= \begin{cases}
        1, & \hat{V}_{m,n}(t) = v\\
        0, & \text{else}
    \end{cases}
\end{align}
\begin{align}
    \label{eq:est_signal_stationary}
    \hat{\pi}_{m,n}^S &= \left[\frac{\sum_{t\in\tau_{m,n}^S}I_{m,n}^S(t,s)}{\#\tau_{m,n}^S} \text{ for $s \in \mathcal{S}$}\right]\\
    I_{m,n}^S(t,s) &= \begin{cases}
        1, & \hat{S}_{m,n}(t) = s\\
        0, & \text{else}
    \end{cases}
\end{align}

Similarly, the motion model state transition probability matrix $\hat{P}_{m,n}^V$ and the signal state transition probability matrix $\hat{P}_{m,n}^S$ can be estimated by node $n$ for each target $m$ by Eq. (\ref{eq:est_motion_transition}) and Eq. (\ref{eq:est_signal_transition}). 

\begin{align}
\label{eq:est_motion_transition}
    \hat{P}_{m,n}^V &= \left[\frac{\sum_{t\in\tau_{m,n}^V}\delta_{m,n}^V(i,j,t)}{\sum_{t\in\tau_{m,n}^V} I_{m,n}^V(j, t)} \text{ for $i,j \in \mathcal{V}$ }\right]\\
    \delta_{m,n}^V(i,j,t) &= \begin{cases}
        1, & \hat{V}_{m,n}(t_0) = i \;\; \& \;\; \hat{V}_{m,n}(t) = j\\
        0, & \text{else}
    \end{cases}\\
    t_0 &= \sup_{t_0\in\tau_{m,n}^V}\{t_0 \text{ s.t. } t_0 < t\}
\end{align}
\begin{align}
\label{eq:est_signal_transition}
    \hat{P}_{m,n}^S &= \left[\frac{\sum_{t\in\tau_{m,n}^S}\delta_{m,n}^S(i,j,t)}{\sum_{t\in\tau_{m,n}^S}I_{m,n}^S(j,t)} \text{ for $i,j \in \mathcal{S}$ }\right]\\
    \delta_{m,n}^S(i,j,t) &= \begin{cases}
        1, & \hat{S}_{m,n}(t_0) = 1 \;\; \& \;\; \hat{V}_{m,n}(t) = j\\
        0, & \text{else}
    \end{cases}\\
    t_0 &= \sup_{t_0\in\tau_{m,n}^V}\{t_0 \text{ s.t. } t_0<t\}
\end{align}

So, in each time step, node $n$ is able to form estimated stationary distributions $\hat{S}_{m,n}$ and $\hat{V}_{m,n}$ for parameters $S_m$ and $V_m$. 

At the end of an epoch, the FC performs a similar process and estimates the distributions using the entire track for each target, as well as the increased number of measurements available to the FC. 
Eq. (\ref{eq:FC_motion_times}) represents the times at which any node chooses the active observation method and Eq. (\ref{eq:FC_signal_times}) represents the set of times at which any node chooses the passive observation method. 
Call the FC estimate for the motion model stationary distribution of target $m$ $\hat{\pi}_m^V$ and call the FC estimate for the signal model stationary distribution of target $m$ $\hat{\pi}_m^S$. 
\begin{equation}
    \label{eq:FC_motion_times}
    \tau_m^V = \bigcup_{n\in\mathcal{N}} \tau_{m,n}^V
\end{equation}
\begin{equation}
    \label{eq:FC_signal_times}
    \tau_m^S = \bigcup_{n\in\mathcal{N}} \tau_{m,n}^S
\end{equation}
These are used to determine the portion of time each target spends in each state. 
Then, we use a modified k-means algorithm to cluster these distributions by similarity. 
The system must also perform model order estimation, as the number of true classes $C$ is not known \emph{a priori}. 
A few different modifications to k-means for distributions have been proposed in the literature, all on the metric used to determine distance between elements: \cite{JMLR:v6:banerjee05b} proposes Bregman divergences, \cite{villani2008optimal} uses the Wasserstein metric (earthmover distance), \cite{e16063273} suggests the $\alpha$-divergence. 
We select the p-Wasserstein metric Eq. (\ref{eq:k-means-metric}) for its simplicity. 
\begin{equation}
    \label{eq:k-means-metric}
    W_p(\hat{\pi}_{m_1}^V, \hat{\pi}_{m_2}^V) = \left(\frac1V \sum_{v=1}^V ||\hat{\pi}_{m_1}^V(v) - \hat{\pi}_{m_2}^V(v)||^p\right)^{1/p}
\end{equation}
This results in $\hat{C}$ different classes. 
Construct sets $\mathcal{M}_{\hat{C}_i}$ containing the target tracks clustered to each class. 
Then, for each estimated class, we call Eq. (\ref{eq:class_motion_stationary}) the class motion model stationary distribution, and similarly for Eq. (\ref{eq:class_signal_stationary}) the class signal model stationary distribution. 
\begin{equation}
\label{eq:class_motion_stationary}
    \overline{\pi}_{\hat{C}_i}^V = \left[\mathlarger{\sum}_{m\in\mathcal{M}_{\hat{C}_i}}\left(\frac{\sum_{t\in\tau_m^V}I_{m}^V(t,v)}{\#\tau_{m}^V}\right) \text{ for $v \in \mathcal{V}$} \right]
\end{equation}
\begin{equation}
\label{eq:class_signal_stationary}
    \overline{\pi}_{\hat{C}_i}^S = \left[ \mathlarger{\sum}_{m\in\mathcal{M}_{\hat{C}_i}} \left( \frac{\sum_{t\in\tau_m^S}I_m^S(t,s)}{\#\tau_m^S} \right) \text{ for $s \in \mathcal{S}$} \right]
\end{equation}

After classes are formed they are distributed to each node. 
Then, new targets may be associated to a class. 
Then, by the following assumption, we can associate targets to classes using either their estimated motion model or signal model stationary distributions. 

\begin{assumption}[Single Family]
\label{ass:single_family}
The targets present in the environment belong to a single target family.     
\end{assumption}

Assumption \ref{ass:single_family} states that all classes may be represented by the same parametres. 
Targets are associated to a class by Eq. (\ref{eq:class_association}), at which point the corresponding tracking filter can be updated to use the motion model probabilities from the class. 

\begin{equation}
\label{eq:class_association}
    C_{m,n} = \min_{c\in\mathcal{C}} \left( W_2(\hat{\pi}_{m,n}^V, \overline{\pi}_{c}^V) + W_2(\hat{\pi}_{m,n}^S, \overline{\pi}_c^S) \right)
\end{equation}

\section{Methods}
\label{sec:methods}

\subsection{Centralized Bandit}
\label{ss:methods_centralized}
We use the common Upper Confidence Bound (\textbf{UCB}) \cite{bandits} \cite{UCB_fischer} formulation, where a single player selects from finitely many actions (``arms'') and observes a corresponding reward. 
Over many iterations, the goal of the player is to maximize the total expected reward. 
To reduce the complexity\footnote{An alternative approach might assign a single bandit algorithm with one arm per combination of node actions, which would total $2^{N}$ arms. Our approach covers the same action space, while reducing the number of arms per bandit algorithm to two. }, we pose the problem with one bandit algorithm per node, which are all evaluated by the FC. 
The algorithms could possibly be implemented by each node, but since the reward function (shown below) requires global information, this approach would require more communication. 

\paragraph{Rewards} The reward for each action is generated by the normalized Shannon entropy of the motion model distribution. 
This value is used since it is constrained to the unit interval and reflects the information content of the motion model distribution: as the distribution of states becomes more flat, the Shannon entropy will increase. 
Other works have often considered target tracking error (via Kalman filter covariance or Bayesian Cram\'er-Rao Lower Bound (\textbf{BCRLB}) as the error for related problems. 
This value has the problem of being very path-dependent; the history of observations can bias the estimated reward substantially. 
Instead, the proposed Shannon entropy reward reflects the estimated class of the target, and is less biased. 
This is particularly useful because as targets become more maneuverable, Kalman filters become less accurate and therefore benefit from more frequent updating \cite{5977487}. 

\begin{equation}
\label{eq:j4_rewards}
	u_n(t) = \frac1{M_n}\sum_{j=1}^{M_n}[\eta(V_j(t)), \eta(S_j(t))] 
\end{equation}

\begin{equation}
\label{eq:shannon}
	\eta(X(t)) = \sum_{i=1}^{n_X} \frac{x_i\log_2(x_i)}{\log_2(n_X)}
\end{equation}

Eq. (\ref{eq:j4_rewards}) shows the reward for selecting either action at node $n$, where $\eta(\cdot)$ represents the normalized Shannon entropy. 
Eq. (\ref{eq:shannon}) shows the Shannon entropy for a distribution $X$ with $n_X$ states $x_i$.  
The reward for selecting the radar action is dependent on the distribution of motion states of covered targets, and the reward for selecting the passive action is dependent on the distribution of the signal states of covered targets. 
So, the reward formulation is dependent on all of the targets viewed by a particular node. 

\paragraph{Mode Selection} Then, in each time step $t$, Eq. (\ref{eq:UCB}) is used to select the mode with index $i$ for node $n$ where $N_t(n)$ is the number of times each mode has been selected before time $t$. 
Table \ref{tab:modes} shows the mode for each index. 
\begin{equation}
\label{eq:UCB}
 \text{Mode}(t) = \argmin_{i\in u_n} \left[u_n + \sqrt{\frac{\log t}{N_t(n)}}\; \right]
\end{equation}
\begin{table}
    \centering
    \caption{CRN Modes}
    \begin{tabular}{c|c}
        Index & Mode \\
        \hline 
        1 & Active Radar\\
        2 & Passive ESM\\[1ex]
    \end{tabular}
    \label{tab:modes}
\end{table}

\subsection{Distributed Approach}
While a centralized technique has the benefit of more well-informed decision-making, it can suffer from latency. 
In other words, the time cost of moving information from the network's edge (the nodes) to the central decision-maker may cause extra error. 
Moving the decision to the edge can mitigate this effect.

We can begin by noting that the \emph{distribution of target age is not stationary in time}. 
\rev{Fig. \ref{fig:number_targets} shows the mean of the target age distribution over 30s. As can be seen, the mean age increases in time. }
When the network is ``switched on'', all targets will be young. 
Since the target model allows for targets to enter and exit the scene, there is a fixed rate at which new targets appear after the beginning of the game. 
Since this is a Poisson process, i.e. the time between new target appearances is exponentially distributed, it is independent of the start time. 
Similarly, target death is exponentially distributed. 
Due to these, the discrepancy in target age distribution is most pronounced at early time steps. 

\begin{figure}
    \centering
    \includegraphics[scale=0.65]{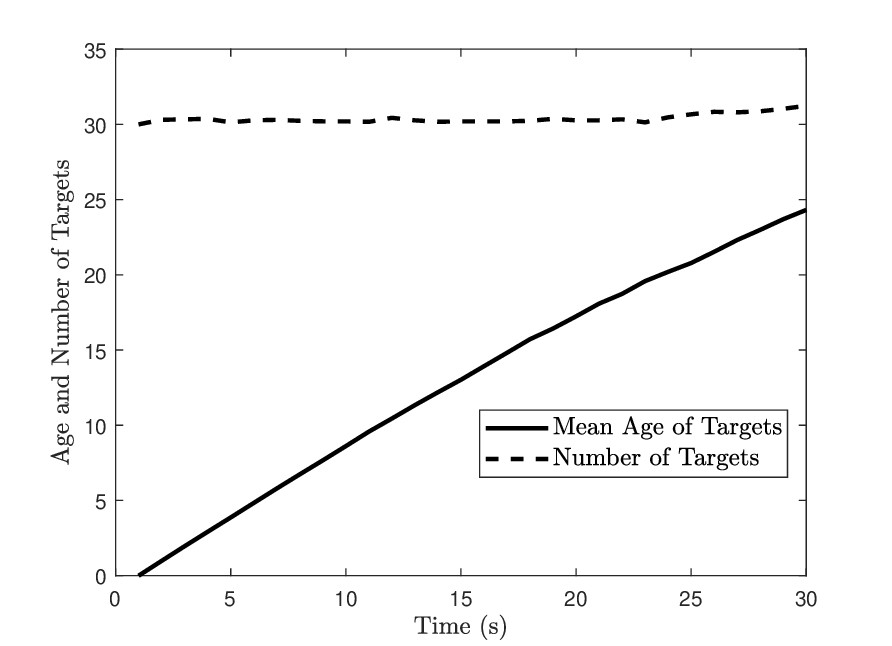}
    \caption{As the scenario ages, the mean target age increases. \rev{The y-axis denotes both the number of targets, as well as the age in seconds. }}
    \label{fig:number_targets}
\end{figure}

The target age distribution is relevant to tracking performance. 
For a given target track, we should expect less error if the target is classified early in the track than if the target is classified later in the track. 

Lastly we must also note that passive signal detection is a more reliable indicator of target class than target motion model, due to the more frequent updating of the signal Markov chain. 

Due to this relationship between target age and tracking performance, we should expect that policies which prioritize passive observation early in each target track will result in lower tracking error. 
So, we present a distributed mode selection technique which is informed by target age.

Eq. (\ref{eq:distributed_utility}) shows the utility function calculated by each node in each time step. 
The age of each track $(\Delta_m)$ is used to weight the Shannon entropy $\eta$ Eq. (\ref{eq:shannon}) of the class stationary distribution for target $m$, if it exists. 
If the target has not been associated with a class, the term $\gamma$ is used to emphasize passive estimation so that the target may be classified. 
Older tracks are weighted less, as the importance of identifying new targets is greater. 
\begin{align}
\label{eq:distributed_utility}
    \mathcal{U}(t) &= \frac{1}{\#\mathcal{M}_n^{(t)}} \sum_{m\in\mathcal{M}_n^{(t)}} \frac{1}{\Delta_m} \left[f(m), 1-f(m)\right]\\
    f(m) &= \begin{cases}
        \eta\left(\overline{\pi}_{C_i}^V\right) , & \exists i \text{ s.t. }m \in\mathcal{M}_{C_i} \\
        \gamma , & \text{else}
    \end{cases}
\end{align}
Then, the selected action at time $t$ is given by Eq. (\ref{eq:distributed_mode}), as randomly sampling the active or passive mode according to the weights in $\mathcal{U}$. 
\begin{equation}
\label{eq:distributed_mode}
    \text{Mode}(t) \overset{{\scriptscriptstyle \operatorname{R}}}{\leftarrow}  \mathcal{U}(t)
\end{equation}



\section{Numerical Simulations}
\label{sec:simulations}
\begin{table}
    \centering
    \caption{Simulation Parameters}
    \begin{tabular}{c|c|c}
        Description & Value & Variable\\[0.5ex] 
        \hline
         Node Density per \qty{}{\km\squared} & 0.2 & $\lambda_n$\\
         Target Density per \qty{}{\km\squared}& 0.3 & $\lambda_m$\\
         Simulated Region & \qty{100}{\km\squared} & $|B|$\\
         Number of Classes & 3 &\\
         Averaged Simulations & 30 &\\ 
         Number of Epochs & 15 &\\ 
         Epoch Duration & 25s &\\ [1ex]
    \end{tabular}    
    \label{tab:j4_params}
\end{table}

We simulate a CRN with the parameters listed in Table \ref{tab:j4_params}. 
In particular, we simulate fifteen epochs. 
After each epoch (i.e., scenario) the FC updates the list of target classes. 
When the game begins, there are no classes, and when it ends the classes should have high accuracy. 
Fig.  \ref{fig:classes} shows that this is the case; class accuracy increases in each epoch. 
Further, Fig.  \ref{fig:classes} also shows that the accuracy with which targets are associated to a class increases in each epoch. 
This, coupled with the result shown in Fig.  \ref{fig:tuned_filters} which shows that tracking accuracy improves when a tuned filter is used, implies that the observed tracking accuracy in the entire network should improve. 
\begin{figure}
    \centering
    \includegraphics[scale=0.65]{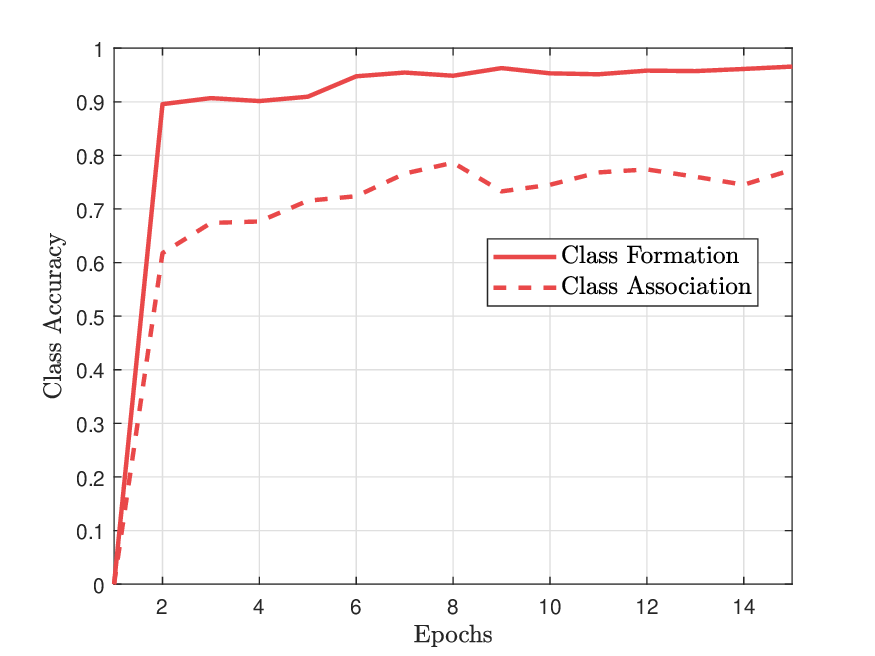}
    \caption{Class formation accuracy is higher than class association accuracy. This is due to the greater number of observations available to the FC during class formation; the nodes must rely on their own measurements to associate targets to classes. }
    \label{fig:classes}
\end{figure}

We compare the distributed and centralized approaches against two baselines: the first selects radar all the time (``Radar Only''), and the second selects radar $80\%$ of the time (``Random $80\%$ Radar''). 
\rev{The empirical value of $80\%$ was selected to approximately match with the observed radar percentages of the Centralized Bandit and Distributed policies. }
We choose these to compare against no mode control, and against a policy which provides equivalent max intercept range and power utilization. 
Fig.  \ref{fig:radar_utilization} shows the percent of time in which radar is selected. 
Passive ESM is selected the other portion of time. 
\begin{figure}
    \centering
    \includegraphics[scale=0.65]{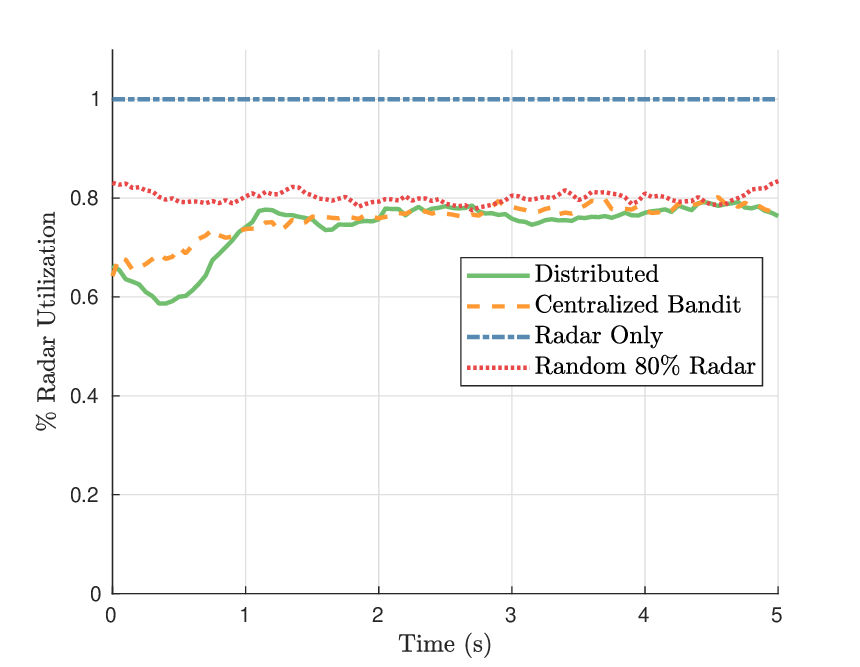}
    \caption{Radar is selected the indicated portion of time, with passive ESM being selected in the complementary portion of time. The centralized and distributed policies are compared against a policy which uses no mode control and only selects radar, and a policy which selects radar randomly $80\%$ of the time. }
    \label{fig:radar_utilization}
\end{figure}

The maximum radar intercept range is the maximum range at which an intercept receiver could detect the radar emissions from a single node. 
Fig. \ref{fig:max_intercept} shows the distribution of maximum intercept range generated by the four different policies. 
Since the centralized, distributed, and $80\%$ radar policies all use radar approximately $80\%$ of the time, they all have similar maximum intercept range distributions. 
Since the policy which only selects radar uses much more radar power, it has a much higher maximum intercept range. 
\begin{figure}
    \centering
    \includegraphics[scale=0.65]{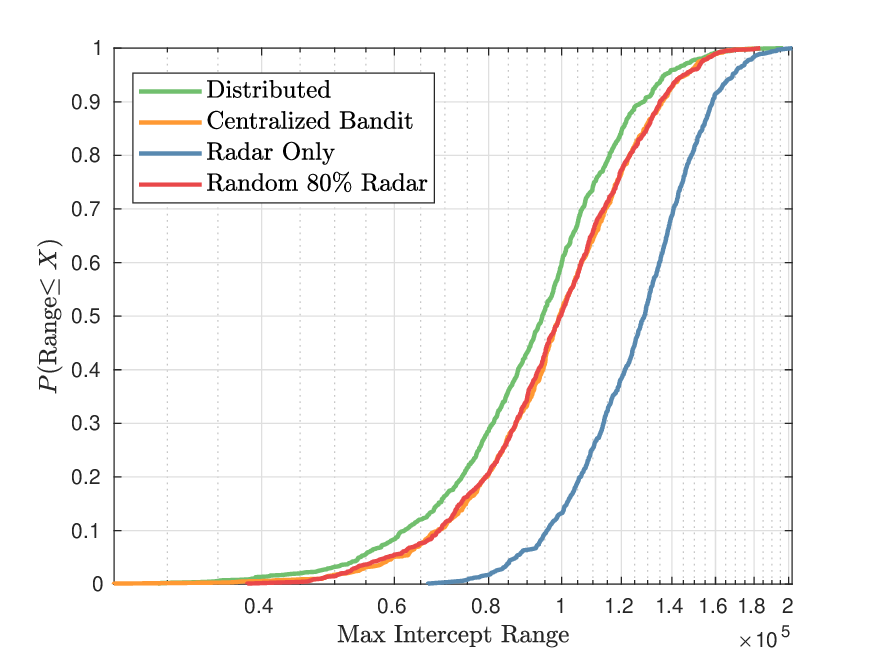}
    \caption{Distribution of the maximum intercept range for the four different policies. }
    \label{fig:max_intercept}
\end{figure}

We present \rev{in Fig. \ref{fig:ECDF}} the tracking error generated by these policies. 
Since the active radar estimation only uses radar and does not utilize target classification, we should expect its performance to indicate the baseline. 
Since the $80\%$ radar policy alternates between active radar and passive ESM, but does not use target classification, we should expect its performance to be worse than the only radar policy. 
Finally, since the centralized and distributed policies both use target classes and passive ESM estimation, we should expect their performance to be very good. 
\rev{We see in Fig. \ref{fig:ECDF} that our expectations hold; the Centralized Bandit (with zero network latency) exhibits lower error probabilities than all other techniques. 
In addition, both the Radar Only and the $80\%$ radar policies underperform the Distributed and Centralized Bandit policies}
\begin{figure}
    \centering
    \includegraphics[scale=0.65]{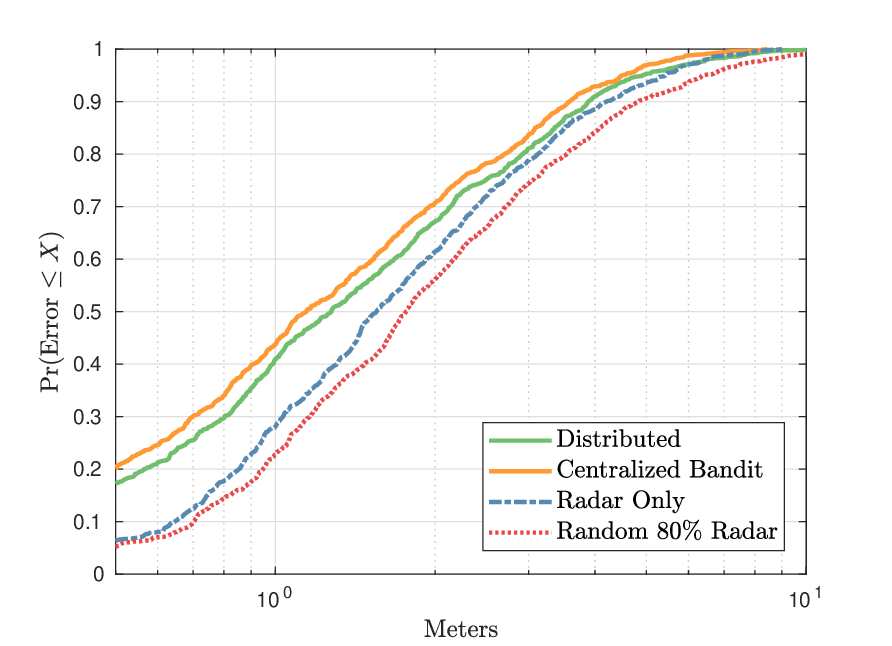}
    \caption{Tracking error distributions for all four policies \rev{with zero network latency}. }
    \label{fig:ECDF}
\end{figure}
\begin{figure}
    \centering
    \includegraphics[scale=0.65]{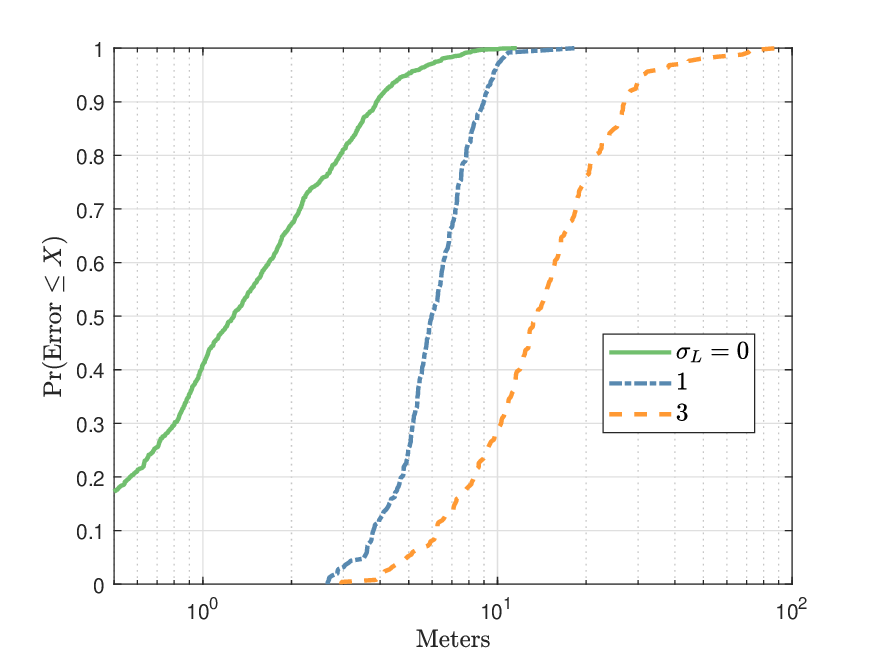}
    \caption{\rev{Tracking error distribution for the Centralized Bandit policy with network latency parameters $\sigma_L$ set to 0, 1, and 3. }}
    \label{fig:latency}
\end{figure}
\rev{Further it can be seen in Fig. \ref{fig:latency} that the impact of increased network latency causes the error experienced by the centralized bandit policy to increase around an order of magnitude. 
This is because the network is unable to process updates in a timely manner, resulting in poor decision-making and tracking. }


\section{Conclusions}
\label{sec:conclusions}

In this work we investigated the capabilities available to a network when multiple modes of operation are present. 
This, along with our previously presented work \cite{howard2023_modecontrolconf}, represents the first contribution towards the field of Cognitive Radar Network (\textbf{CRN}) mode control. 
In particular, we examine the case where the nodes of a CRN have, in addition to active radar, the ability to conduct passive signal parameter estimation. 
In each of many time steps, every node in the CRN can operate in one of these two modes. 
When conducting radar observation, a node provides to the \rev{Fusion Center (\textbf{FC})} an estimate of the position and velocity of all targets within its range. 
When conducting passive signal parameter estimation, a node provides to the FC an estimate of the signal emissions from all targets within its range. 
The passive measurements are associated with targets via direction of arrival estimation. 

In addition to these direct observations, the FC maintains records of the motion model (i.e., constant velocity, constant turn, etc. ) and the history of signal emissions of each target. 
Modeling both of these as Markov processes, the FC estimates the transition probabilities for each of these parameters over time. 
On fixed intervals (``epochs’’), the FC then clusters these targets into ``classes’’ which contain targets with similar behavior. 
Finally, using these constructed target classes, the FC is able to estimate the class of future targets in order to determine their likely behavior. 
In this way, the motion model of targets is able to be estimated using \emph{passive} observation, as the class of a target is dependent on both signal and motion characteristics. 
So, the FC is able to trade between active radar observation and passive signal parameter estimation over time. 

We show that the use of this estimation technique can 1) reduce the effective radiated power of the network by decreasing the proportion of time radar is performed and 2) increase the target tracking accuracy of the network by leveraging ``prior'' information on the targets. 
We demonstrate that passive signal parameter estimation is constrained by the maximum detectable range and probability of intercept of the targets emissions.
In addition, we contribute mathematical analysis of the class formation technique. 

We demonstrated that the distributed estimation technique can mitigate the effects of network latency by moving the decision-making to the edge of the network. 
The centralized technique, however, benefits from more complete knowledge of the targets.


\bibliographystyle{IEEEtran}
\bibliography{bibli}

\begin{thebibliography}{10}
\providecommand{\url}[1]{#1}
\csname url@samestyle\endcsname
\providecommand{\newblock}{\relax}
\providecommand{\bibinfo}[2]{#2}
\providecommand{\BIBentrySTDinterwordspacing}{\spaceskip=0pt\relax}
\providecommand{\BIBentryALTinterwordstretchfactor}{4}
\providecommand{\BIBentryALTinterwordspacing}{\spaceskip=\fontdimen2\font plus
\BIBentryALTinterwordstretchfactor\fontdimen3\font minus
  \fontdimen4\font\relax}
\providecommand{\BIBforeignlanguage}[2]{{%
\expandafter\ifx\csname l@#1\endcsname\relax
\typeout{** WARNING: IEEEtran.bst: No hyphenation pattern has been}%
\typeout{** loaded for the language `#1'. Using the pattern for}%
\typeout{** the default language instead.}%
\else
\language=\csname l@#1\endcsname
\fi
#2}}
\providecommand{\BIBdecl}{\relax}
\BIBdecl

\bibitem{howard2023_modecontrolconf}
W.~W. Howard, S.~R. Shebert, B.~H. Kirk, and R.~M. Buehrer, ``Mode selection
  and target classification in cognitive radar networks,'' \emph{arXiv preprint
  arXiv:2310.17006}, 2023.

\bibitem{howard2022_MMABjournal}
W.~W. Howard, A.~F. Martone, and R.~M. Buehrer, ``Distributed online learning
  for coexistence in cognitive radar networks,'' \emph{IEEE Transactions on
  Aerospace and Electronic Systems}, pp. 1--14, 2022.

\bibitem{howard2023_hybridjournal}
W.~W. Howard and R.~M. Buehrer, ``Hybrid cognition for target tracking in
  cognitive radar networks,'' 2023.

\bibitem{haykin2005}
S.~{Haykin}, ``Cognitive radar networks,'' in \emph{1st IEEE International
  Workshop on Computational Advances in Multi-Sensor Adaptive Processing,
  2005.}, 2005, pp. 1--3.

\bibitem{thornton2022_universaljournal}
C.~E.~Thornton, R.~M. Buehrer, H.~S.~Dhillon, and A.~F.~Martone, ``Universal
  learning waveform selection strategies for adaptive target tracking,''
  \emph{IEEE Transactions on Aerospace and Electronic Systems}, pp. 1--17,
  2022.

\bibitem{9178313}
C.~E. {Thornton}, M.~A. {Kozy}, R.~M. {Buehrer}, A.~F. {Martone}, and K.~D.
  {Sherbondy}, ``Deep reinforcement learning control for radar detection and
  tracking in congested spectral environments,'' \emph{IEEE Transactions on
  Cognitive Communications and Networking}, pp. 1--1, 2020.

\bibitem{howard2021_multiplayerconf}
W.~W. Howard, C.~E. Thornton, A.~F. Martone, and R.~M. Buehrer, ``Multi-player
  bandits for distributed cognitive radar,'' in \emph{2021 IEEE Radar
  Conference (RadarConf21)}.\hskip 1em plus 0.5em minus 0.4em\relax IEEE, 2021,
  pp. 1--6.

\bibitem{howard2023_timelyjournal}
W.~W. Howard, A.~F. Martone, and R.~M. Buehrer, ``Timely target tracking:
  Distributed updating in cognitive radar networks,'' \emph{IEEE Transactions
  on Radar Systems}, vol.~2, pp. 318--332, 2024.

\bibitem{thornton2020efficient}
C.~E. Thornton, R.~M. Buehrer, and A.~F. Martone, ``Efficient online learning
  for cognitive radar-cellular coexistence via contextual thompson sampling,''
  in \emph{GLOBECOM 2020 - 2020 IEEE Global Communications Conference}, 2020,
  pp. 1--6.

\bibitem{Martone_CRN_loop}
A.~F. Martone, K.~D. Sherbondy, J.~A. Kovarskiy, B.~H. Kirk, R.~M. Narayanan,
  C.~E. Thornton, R.~M. Buehrer, J.~W. Owen, B.~Ravenscroft, S.~Blunt,
  A.~Egbert, A.~Goad, and C.~Baylis, ``Closing the loop on cognitive radar for
  spectrum sharing,'' \emph{IEEE Aerospace and Electronic Systems Magazine},
  vol.~36, no.~9, pp. 44--55, 2021.

\bibitem{6807568}
G.~Alirezaei, M.~Reyer, and R.~Mathar, ``Optimum power allocation in sensor
  networks for passive radar applications,'' \emph{IEEE Transactions on
  Wireless Communications}, vol.~13, no.~6, pp. 3222--3231, 2014.

\bibitem{Bogler1987}
P.~L. Bogler, ``{Shafer-Dempster Reasoning with Applications to Multisensor
  Target Identification Systems},'' \emph{IEEE Transactions on Systems, Man,
  and Cybernetics}, vol.~17, no.~6, pp. 968--977, 1987.

\bibitem{Hong1993}
L.~Hong and A.~Lynch, ``{Recursive Temporal-Spatial Information Fusion with
  Applications to Target Identification},'' \emph{IEEE Transactions on
  Aerospace and Electronic Systems}, vol.~29, no.~2, pp. 435--445, 1993.

\bibitem{Lei2020}
Z.~Lei, P.~Cui, and Y.~Huang, ``{Multi-platform and Multi-sensor Data Fusion
  Based on D-S Evidence Theory},'' in \emph{2020 IEEE 3rd International
  Conference on Computer and Communication Engineering Technology (CCET)},
  2020, pp. 6--9.

\bibitem{Li2022}
R.~Li, Y.~Zhang, and J.~Sun, ``{Active and Passive Radar Target Fusion
  Recognition Method Based on Bayesian Network},'' in \emph{2022 15th
  International Congress on Image and Signal Processing, BioMedical Engineering
  and Informatics (CISP-BMEI)}, 2022, pp. 1--5.

\bibitem{Challa2001}
S.~Challa and G.~Pulford, ``{Joint target tracking and classification using
  radar and ESM sensors},'' \emph{IEEE Transactions on Aerospace and Electronic
  Systems}, vol.~37, no.~3, pp. 1039--1055, 2001.

\bibitem{Cao2018}
W.~Cao, J.~Lan, and X.~R. Li, ``{Extended Object Tracking and Classification
  Using Radar and ESM Sensor Data},'' \emph{IEEE Signal Processing Letters},
  vol.~25, no.~1, pp. 90--94, 2018.

\bibitem{8386661}
F.~Liu, L.~Zhou, C.~Masouros, A.~Li, W.~Luo, and A.~Petropulu, ``Toward
  dual-functional radar-communication systems: Optimal waveform design,''
  \emph{IEEE Transactions on Signal Processing}, vol.~66, no.~16, pp.
  4264--4279, 2018.

\bibitem{8114253}
P.~Kumari, J.~Choi, N.~González-Prelcic, and R.~W. Heath, ``Ieee
  802.11ad-based radar: An approach to joint vehicular communication-radar
  system,'' \emph{IEEE Transactions on Vehicular Technology}, vol.~67, no.~4,
  pp. 3012--3027, 2018.

\bibitem{10038921}
S.~A. Ford and M.~Ritchie, ``Cognitive radar mode control: a comparison of
  different reinforcement learning algorithms,'' in \emph{International
  Conference on Radar Systems (RADAR 2022)}, vol. 2022, 2022, pp. 107--112.

\bibitem{rs15163977}
\BIBentryALTinterwordspacing
T.~Pietkiewicz, ``Fusion of identification information from esm sensors and
  radars using dezert–smarandache theory rules,'' \emph{Remote Sensing},
  vol.~15, no.~16, 2023. [Online]. Available:
  \url{https://www.mdpi.com/2072-4292/15/16/3977}
\BIBentrySTDinterwordspacing

\bibitem{8943388}
Z.~Wang, Y.~Wu, and Q.~Niu, ``Multi-sensor fusion in automated driving: A
  survey,'' \emph{IEEE Access}, vol.~8, pp. 2847--2868, 2020.

\bibitem{7979175}
M.~L. Fung, M.~Z.~Q. Chen, and Y.~H. Chen, ``Sensor fusion: A review of methods
  and applications,'' in \emph{2017 29th Chinese Control And Decision
  Conference (CCDC)}, 2017, pp. 3853--3860.

\bibitem{1271397}
S.~Julier and J.~Uhlmann, ``Unscented filtering and nonlinear estimation,''
  \emph{Proceedings of the IEEE}, vol.~92, no.~3, pp. 401--422, 2004.

\bibitem{1710358}
B.-N. Vo and W.-K. Ma, ``The gaussian mixture probability hypothesis density
  filter,'' \emph{IEEE Transactions on Signal Processing}, vol.~54, no.~11, pp.
  4091--4104, 2006.

\bibitem{5259179}
K.~Panta, D.~E. Clark, and B.-N. Vo, ``Data association and track management
  for the gaussian mixture probability hypothesis density filter,'' \emph{IEEE
  Transactions on Aerospace and Electronic Systems}, vol.~45, no.~3, pp.
  1003--1016, 2009.

\bibitem{4516991}
K.~Punithakumar, T.~Kirubarajan, and A.~Sinha, ``Multiple-model probability
  hypothesis density filter for tracking maneuvering targets,'' \emph{IEEE
  Transactions on Aerospace and Electronic Systems}, vol.~44, no.~1, pp.
  87--98, 2008.

\bibitem{markov_model_book}
R.~Howard, \emph{Dynamic Probabilistic Systems, Volume I: Markov Models}, ser.
  Dover Books on Mathematics.\hskip 1em plus 0.5em minus 0.4em\relax Dover
  Publications, 2007.

\bibitem{norris1998markov}
J.~R. Norris, \emph{Markov Chains}, ser. Cambridge Series in Statistical and
  Probabilistic Mathematics.\hskip 1em plus 0.5em minus 0.4em\relax Cambridge
  University Press, 1998.

\bibitem{ristic2003beyond}
B.~Ristic, S.~Arulampalam, and N.~Gordon, \emph{Beyond the Kalman Filter:
  Particle Filters for Tracking Applications}.\hskip 1em plus 0.5em minus
  0.4em\relax Artech House, 2003.

\bibitem{howard2023_timelyconf}
W.~W. Howard, C.~E. Thornton, and R.~M. Buehrer, ``Timely target tracking in
  cognitive radar networks,'' 2023.

\bibitem{haenggi}
M.~Haenggi, \emph{Stochastic Geometry for Wireless Networks}, ser. Stochastic
  Geometry for Wireless Networks.\hskip 1em plus 0.5em minus 0.4em\relax
  Cambridge University Press, 2013.

\bibitem{8464057}
A.~{Munari}, L.~{Simić}, and M.~{Petrova}, ``Stochastic geometry interference
  analysis of radar network performance,'' \emph{IEEE Communications Letters},
  vol.~22, no.~11, pp. 2362--2365, 2018.

\bibitem{6507656}
B.-T. Vo and B.-N. Vo, ``Labeled random finite sets and multi-object conjugate
  priors,'' \emph{IEEE Transactions on Signal Processing}, vol.~61, no.~13, pp.
  3460--3475, 2013.

\bibitem{granstrom2017extended}
K.~Granstrom, M.~Baum, and S.~Reuter, ``Extended object tracking: Introduction,
  overview and applications,'' 2017.

\bibitem{6178085}
B.~Ristic, D.~Clark, B.-N. Vo, and B.-T. Vo, ``Adaptive target birth intensity
  for {PHD} and {CPHD} filters,'' \emph{IEEE Transactions on Aerospace and
  Electronic Systems}, vol.~48, no.~2, pp. 1656--1668, 2012.

\bibitem{6914549}
W.~A. Jerjawi, Y.~A. Eldemerdash, and O.~A. Dobre, ``Second-order
  cyclostationarity-based detection of {LTE} {SC-FDMA} signals for cognitive
  radio systems,'' \emph{IEEE Transactions on Instrumentation and Measurement},
  vol.~64, no.~3, pp. 823--833, 2015.

\bibitem{5757489}
A.~Al-Habashna, O.~A. Dobre, R.~Venkatesan, and D.~C. Popescu, ``Joint signal
  detection and classification of mobile {WiMAX} and {LTE} {OFDM} signals for
  cognitive radio,'' in \emph{2010 Conference Record of the Forty Fourth
  Asilomar Conference on Signals, Systems and Computers}, 2010, pp. 160--164.

\bibitem{9101176}
R.~Wiley, \emph{ELINT: The Interception and Analysis of Radar Signals}.\hskip
  1em plus 0.5em minus 0.4em\relax Artech House, 2006.

\bibitem{ANTONIOU200272}
I.~Antoniou, V.~Ivanov, V.~V. Ivanov, and P.~Zrelov, ``On the log-normal
  distribution of network traffic,'' \emph{Physica D: Nonlinear Phenomena},
  vol. 167, no.~1, pp. 72--85, 2002.

\bibitem{8691027}
H.~Volos, T.~Bando, and K.~Konishi, ``Latency modeling for mobile edge
  computing using lte measurements,'' in \emph{2018 IEEE 88th Vehicular
  Technology Conference (VTC-Fall)}, 2018, pp. 1--5.

\bibitem{JMLR:v6:banerjee05b}
A.~Banerjee, S.~Merugu, I.~S. Dhillon, and J.~Ghosh, ``Clustering with bregman
  divergences,'' \emph{Journal of Machine Learning Research}, vol.~6, no.~58,
  pp. 1705--1749, 2005.

\bibitem{villani2008optimal}
C.~Villani, \emph{Optimal Transport: Old and New}, ser. Grundlehren der
  mathematischen Wissenschaften.\hskip 1em plus 0.5em minus 0.4em\relax
  Springer Berlin Heidelberg, 2008.

\bibitem{e16063273}
F.~Nielsen, R.~Nock, and S.~Amari, ``On clustering histograms with k-means by
  using mixed $\alpha$-divergences,'' \emph{Entropy}, vol.~16, no.~6, pp.
  3273--3301, 2014.

\bibitem{bandits}
T.~{Lattimore} and C.~{Szepesvari}, \emph{Bandit Algorithms}.\hskip 1em plus
  0.5em minus 0.4em\relax Cambridge University Press, 2020.

\bibitem{UCB_fischer}
P.~Auer, N.~Cesa-Bianchi, and P.~Fischer, ``Finite-time analysis of the
  multiarmed bandit problem,'' \emph{Machine Learning}, vol.~47, pp. 235--256,
  05 2002.

\bibitem{5977487}
M.~Silbert, S.~Sarkani, and T.~Mazzuchi, ``Comparing the state estimates of a
  {K}alman filter to a perfect {IMM} against a maneuvering target,'' in
  \emph{14th International Conference on Information Fusion}, 2011, pp. 1--5.

\end{thebibliography}

\end{document}